\documentclass[superscriptaddress,twocolumn]{revtex4-1}

\usepackage{amsmath,amsfonts,amssymb,amsthm,epsfig,times,bbm}
\usepackage[usenames,dvipsnames]{color}
\usepackage{graphicx}
\usepackage{multirow}
\usepackage{mathtools}

\mathtoolsset{showonlyrefs=true}

\definecolor{jens}{rgb}{.2,0.7,.9}

\definecolor{mathis}{rgb}{.9,.0,.9}

\definecolor{albert}{rgb}{0.8,0.1,0.89}

\newcommand{\cc}{\mathbbm{C}}

\newcommand{\rr}{\mathbbm{R}}
\newtheorem{thm}{Theorem}
\newtheorem{lm}{Lemma}
\newtheorem{df}[lm]{Definition}
\newtheorem{dfrestated}{Definition}
\newtheorem{dflm}[lm]{Definition \& Lemma}
\newtheorem*{lm*}{Lemma}

\newtheorem*{prop*}{Proposition}

\newtheorem{co}[thm]{Corollary}
\newtheorem*{co*}{Corollary}

\newcommand{\ket}[1]{\left \vert #1 \right \rangle}
\newcommand{\bra}[1]{\left \langle #1 \right \vert}
\newcommand{\ketbra}[2]{\left \vert #1 \right \rangle \! \!\left \langle #2 \right \vert}
\newcommand{\braket}[2]{\langle #1 \vert #2 \rangle}

\newcommand\Abs[1]{\ensuremath{\left| #1\right|}}

\def\idty{\mathbbm{1}}
\DeclareMathOperator{\vspan}{span}
\newcommand\intd{\mathrm d}
\newcommand\dist[1][(A,B)]{\mathrm{d}#1}
\newcommand\idos{\Theta}

\begin{document}

\title{Many-body localisation implies that eigenvectors are matrix-product states}
\author{M.\ Friesdorf}
\affiliation{Dahlem Center for Complex Quantum Systems, Freie Universit{\"a}t Berlin, 14195 Berlin, Germany}
\author{A.\ H.\ Werner}
\affiliation{Dahlem Center for Complex Quantum Systems, Freie Universit{\"a}t Berlin, 14195 Berlin, Germany}
\author{W.\ Brown}
\affiliation{Computer Science Department, University College London, London WC1E 6BT}
\author{V.\ B.\ Scholz}
\affiliation{Institute for Theoretical Physics, ETH Zurich, Wolfgang-Pauli-Str. 27, 8093 Zurich, Switzerland}
\author{J.\ Eisert}
\affiliation{Dahlem Center for Complex Quantum Systems, Freie Universit{\"a}t Berlin, 14195 Berlin, Germany}

\date{\today}

\begin{abstract}
The phenomenon of many-body localisation received  a lot of attention recently, both for its implications in condensed-matter physics
of allowing systems to be an insulator even at non-zero temperature as well as in the context of the foundations of quantum statistical mechanics,
providing examples of systems showing the absence of thermalisation following out-of-equilibrium dynamics.
In this work, we establish a novel link between dynamical properties
-- a vanishing group velocity and the absence of transport -- with entanglement properties
of individual eigenvectors.
Using Lieb-Robinson bounds and filter functions, we prove rigorously under simple assumptions on the spectrum
 that if a system shows strong dynamical localisation,
 all of its many-body eigenvectors have clustering correlations.
 In one dimension this implies directly an entanglement area law,
 hence the eigenvectors can be approximated by matrix-product states.
We also show this statement for parts of the spectrum, allowing for the existence of a mobility edge above which transport is possible.
\end{abstract}
\maketitle

The concept of disorder induced localisation has been introduced in the seminal work by Anderson \cite{Anderson}
who captured the mechanism responsible for the
absence of diffusion of waves in disordered media.
This mechanism is specifically well understood in the single-particle case,
where one can show that in the presence of a suitable random potential,
all eigenfunctions are exponentially localised \cite{Anderson}.
In addition to this spectral characterisation of localisation there is a notion of dynamical localisation,
which requires that the transition amplitudes between lattice sites decay exponentially \cite{klein2007mulstiscale,stolz2011}.

Naturally, there is a great interest in extending these results to the many-body setting \cite{HuseReview, Bauer}.
In the case of integrable systems that can be mapped to free fermions, such as the XY chain,
results on single particle localisation can be applied directly \cite{Burrell,Sims_zeroLR}.
A far more intricate situation arises in interacting systems.
Such {\it many-body localisation} \cite{Basko,Basko2} has received an enormous attention recently.
In terms of condensed-matter physics, this phenomenon allows for systems to remain
an insulator even at non-zero temperature \cite{Aleiner_Disorder}, in principle even at infinite temperature \cite{Oganesyan}.
In the {\it foundations of statistical mechanics}, such many-body localised systems provide examples of systems that fail to thermalise.
When pushed out of equilibrium, signatures of the initial condition will locally be measurable even after long times,
in contradiction to what one might expect from quantum statistical mechanics \cite{PalHuse,HuseReview,Integrable}.

Despite great efforts to approach the phenomenon of many-body localisation, many aspects are not fully understood and
a comprehensive definition is still lacking.
Similar to the case of single-particle Anderson localisation, there are two complementary approaches to capture the phenomenon.
On the one hand probes involving real-time dynamics \cite{1404.5216,Pollmann_unbounded} have been discussed,
showing excitations ``getting stuck'', or seeing suitable signatures in density-auto-correlation functions \cite{Prozen_localisation},
leading to a dynamic reading of the phenomena.
On the other hand, statistics of energy levels has been considered as an indicator \cite{Oganesyan} as well as a lack of entanglement in the
eigenbasis \cite{Bauer} and a resulting violation of the eigenstate thermalisation hypothesis (ETH) \cite{Srednicki94,Deutsch91,Rigol_etal08}.
In such a static description of many-body localisation, it has for example been suggested to define
many-body localised states as those states that can be prepared from product states or Slater determinants of single-particle localised states with a finite-depth local unitary transformation.
Accordingly, a many-body localised phase is one where most eigenvectors have this property \cite{Bauer}.

It seems fair to say that the connection between these static and dynamic readings is still unclear.
In this work, we make a substantial first step in relating these approaches:
If the evolution of an interacting many-body system
is dynamically localising (and its spectrum is generic),
then all eigenvectors have exponentially clustering correlations in general and,
in one dimension, satisfy an area law and can be efficiently written as a matrix-product state (MPS)
\cite{FCS,MPSSurvey} with a low bond dimension.
Our proof uses advanced mathematical tools developed for the analysis of interacting many-body systems,
such as energy filtering and does not rely on the precise details of the considered model. We show that
short ranged correlations in individual energy eigenstates necessarily follow from absence of transport
in localising models.

{\it Setting.}
For simplicity, we will state all our results in the language of spin systems of local dimension $d$,
but they can equivalently be derived for
fermionic systems.
We consider local Hamiltonians
$ H = \sum_{j} h_j$,
where each interaction term $h_j$ is supported on finitely many sites of a lattice with $N$ vertices,
and the physical system associated with each site is finite-dimensional.
We make for most of the argument no assumptions concerning the dimensionality or the kind of the lattice.
In what follows, we will consider observables $A$ and
$B$ that are supported on finitely many sites on the lattice. There is a natural distance $\dist$
in the lattice related to the length of the shortest path of interactions connecting the supports of $A$ and $B$.
For convenience of notation, we will choose $\|A\| = 1 = \|B\|$ from here on, which does not restrict generality
\cite{HastingsKoma06,Decay2,Quench,BravyiHastingsVerstraete06,NachtergaeleOgataSims06}.
For any local Hamiltonian of this sort there is an upper bound $v>0$ to the largest group velocity,
referred to as Lieb-Robinson velocity.
There are several possible formulations of such a bound (see Appendix \ref{LR} for a detailed discussion).
One convenient and common way is to compare the time-evolution of a local observable
$A(t) = e^{i t H} A e^{-i t H}$
under the full Hamiltonian $H$ with its time-evolution under a truncated Hamiltonian $H_A^l$
that only includes  interactions $h_j$ contained in a region of distance no more than $l$ from
the support of $A$.
The Lieb-Robinson theorem then states that there exists a constant $c>0$ and a velocity $v>0$ such that
\begin{align}
  \label{eq:LR}
  \|A(t) - e^{i t H_A^l} A e^{-i t H_A^l}\|
  \leq c e^{- \mu (l - v t)}
\end{align}
is true for all times $t\geq 0$.
Such a bound limits the \emph{speed of information propagation}:
In any such lattice with short-ranged interactions, all interactions are causal in this sense.
A related, but in the case of $v=0$ weaker estimate is given by the following commutator bound
(see also \ref{LR})
\begin{align}
  \| [A , B(t) ] \| & \leq C e^{- \mu (\dist - v t)} .
  \label{eq:LRcomm}
\end{align}

Our results will rely on two generic and natural assumptions on the spectrum of the Hamiltonian.
To state those, we express the many-body Hamiltonian $H$ in its eigenbasis as
\begin{equation}
  H = \sum_k E_k \ketbra{k}{k}.
\end{equation}
\begin{itemize}
  \item \textit{Non-degenerate energies (\textbf{AI}):}
    The energies of the full Hamiltonian are assumed to be non-degenerate.
    The smallest gap between these energies will be called $\gamma$.
  \item \textit{Locally independent gaps (\textbf{AII}):}
    The energies of reduced Hamiltonians $H_A$ and $H_B$ which include all interactions inside
    rectangular regions $A$ and $B$ respectively, are assumed to be non-degenerate
    when viewing them as operators on their respective
    truncated Hilbert spaces $\mathcal{H}_A$ and $\mathcal{H}_B$.
    The smallest gap will be called $\tilde{\gamma}$.
    Moreover, the gaps of these Hamiltonians need to be locally independent with respect to
    each other, in the sense that
    \begin{align}
      E_a - E_{a'} = E_b - E_{b'} \Rightarrow a=a', b= b' ,
    \end{align}
    where $a,a'$ and $b,b'$ label the eigenvalues of $H_A$ and $H_B$ respectively.
    The smallest difference of these gaps will be called $\eta$.
  \item \textit{Non-symmetric gaps (\textbf{AIII}):}
    For all eigenvalues $E_k$ of the full Hamiltonian the spectrum is
    no-where symmetric in the sense that
    \begin{align}
      \min_{r \neq s} \bigl| |E_r - E_k| - |E_s - E_k| \bigl| \geq \zeta > 0 \; \forall k.
    \end{align}
\end{itemize}
Similar reasonable conditions have already been considered in the context of equilibration of closed systems
and are assumed to hold when a small amount of random noise is added to a local Hamiltonian \cite{Linden_etal09}.
Especially in the context of localising systems which are typically based on large random terms in the Hamiltonian,
these assumptions are very natural indeed.
Our assumptions are, however, considerably weaker than assuming fully non-degenerate energy gaps \cite{Linden_etal09}.
Another quantity that will be important later on is the number of states up to energy $E$ given by
$\idos(E) := \sum_{{l, E_l \leq E}} 1 $.

{\it Notions of dynamical localisation.}
One defining feature of localisation is that
the system shows no transport of information.
There are many possible readings of this.
The first and strongest is that the time evolution can be truncated to a finite region independent of time
and is directly connected to an absence of thermalisation in the model (see Appendix \ref{absence}).

\begin{df}[Strong dynamical localisation]
  \label{df_dynamical_localisation}
  A Hamiltonian $H$ exhibits strong dynamical localisation iff its time evolution satisfies
  \begin{align}\label{eq:str_dyn_loc}
    \|A(t) - e^{i t H_A^l} A e^{-i t H_A^l}\|
    \leq  c_\mathrm{loc} e^{- \mu l} \; ,
  \end{align}
  for a suitable $\mu>0$,
  where $A$ is an arbitrary local observable, $c_\mathrm{loc}>0$ a constant independent of the
  system size and $H_A^l$ denotes a Hamiltonian which includes all interactions
  contained in a region of distance no more than $l$ from the support of $A$.
\end{df}

This corresponds to setting the Lieb-Robinson velocity $v$ to zero in Eq.\ \eqref{eq:LRcomm}.
A potential candidate for such a model is given by free fermionic systems with strong local disorder.
There, it is known that transport is strongly suppressed and it is still being debated to what extent they satisfy
zero velocity Lieb-Robinson bounds \cite{Burrell,HamzaSimsStolz12}.
This definition also implies that information cannot be transported by encoding it locally in the chosen Hamiltonian
(see Appendix \ref{LR} for details).
As such, it does not allow for a growth of entanglement
entropies logarithmically in time, a phenomenon
that has been numerically observed in spin chains with random onsite noise \cite{Pollmann_unbounded}.

In a second step, we relax our dynamical assumption
to a weaker variant of zero velocity Lieb-Robinson bounds and restricting it to a suitable subspace.
It is natural to assume that the transport properties of an electronic system are closely connected to the
energy available in the system. In fact, often the notion of a \emph{mobility edge} is introduced
\cite{1001.5280,HamzaSimsStolz12}.
The intuition is that for all energies below this mobility edge, transport is fully blocked,
while it might be possible for higher energies.
The following definition is making this precise and allows for transport in highly
excited sectors.
\begin{df}[Mobility edge]
  \label{def:subsp_loc}
  A Hamiltonian $H$ is said to have a mobility edge at energy $E_{\mathrm{mob}}$
  iff its time evolution satisfies for all times $t$,
    $\forall \rho\in\vspan\{\ketbra{l}{k}: \; E_l,E_k \leq E_{\mathrm{mob}}\}$
    \begin{align}
      \label{eq:zeroLRcomm}
      | \mathrm{tr} \left(\rho [A(t),B]\right)\vert
      \leq \min(t,1) c_\mathrm{mob} e^{-\mu \dist}\; ,
    \end{align}
    with constants $c_\mathrm{mob}, \mu>0$ independent of $t$.
    That is, all transport is suppressed for states supported only on the low-energy sector
    below energy $E_{\mathrm{mob}}$.
\end{df}

This definition relies on a rather weak version of zero velocity Lieb-Robinson bounds (see Appendix \ref{LR}).
While our results will be applicable to all eigenvectors below the mobility edge,
they can naturally also be used to capture the ground state,
yielding an extension of the previous results \cite{HamzaSimsStolz12}.

In experiments, a natural setting for dynamical localisation
is one where a system is excited locally.
For example, in case of a ground state vector $\ket{0}$,
this corresponds to the action of a unitary operator $U$ with local support
$ \ket{\psi} = U \ket{0} = e^{- i G} \ket{0} $
where $G$ denotes its generator.
In this setting, suppression of transport means that the excitation cannot 
be measured at far distances, even for long times, in the sense that
\begin{align}
  |\bra{0}A(t)-e^{i G} A(t) e^{-i G} \ket{0}| \leq \min(t,1)\Vert{G}\Vert C e^{-\mu d(A,U)},
\end{align}
for arbitrary observables $A$.
Here $C, \mu > 0$ are constants independent of the time $t$ and the system size.
If excitations are stuck in the above sense in an entire certain energy subspace,
then this implies our zero velocity Lieb-Robinson estimate in Eq.\ \eqref{eq:zeroLRcomm},
thus relating Definition \ref{def:subsp_loc} to a quantity that can in principle be experimentally measured (see Appendix \ref{LR}).

{\it Main result.}
We are now in the position to state our main result.
It provides a clear link between the dynamical properties
of the system and the static correlation behaviour of the corresponding eigenvectors and
can be applied to quantum systems of arbitrary dimension. In one dimension it directly implies
an area law and the existence of an approximating MPS (see Corollary \ref{mps_theorem}).

\begin{thm}[Clustering of correlations of eigenvectors]
  \label{main_theorem}
  The dynamical properties of a local Hamiltonian imply clustering correlations
  of its eigenvectors in the following way.
  \begin{itemize}
    \item[\textbf{a}]
      If the Hamiltonian shows strong dynamical localisation
      and its spectrum fulfils assumptions \textbf{AI} and \textbf{AII}
      then all its eigenvectors $\ket k$ have exponentially clustering correlations, i.e. 
      \begin{align}
        |\bra{k} A B \ket{k} - \bra{k} A \ketbra{k}{k} B \ket{k}|
        \leq 4 c_\mathrm{loc} e^{-{\mu} \dist /2} .
      \end{align}
    \item[\textbf{b}]
      If the Hamiltonian has a mobility edge at energy $E_\mathrm{mob}$
      and its spectrum fulfils assumptions \textbf{AI} and \textbf{AIII},
      then all eigenvectors $\ket{k}$  up to that energy $E_\mathrm{mob}$ cluster exponentially,
      in the sense that for all $\kappa>0$
      \begin{eqnarray}
        \nonumber
        &&\Abs{\bra{k} A B \ket{k} - \bra{k} A\ketbra{k}{k} B \ket{k}}\\
        &\leq&
        \left(12 \pi \idos(E_k + \kappa) c_\mathrm{mob} + \ln \frac{\pi \mu \dist e^{4 + 2 \pi}}{\kappa^2}
        \right) \frac{e^{-\mu \dist}}{2 \pi} \;,
      \end{eqnarray}
      where $\idos(E)$ is the number of eigenstates up to energy $E$ and $\kappa$ can be chosen arbitrarily to optimise the bound.
  \end{itemize}
\end{thm}

\begin{figure*}
  \includegraphics[width=\textwidth]{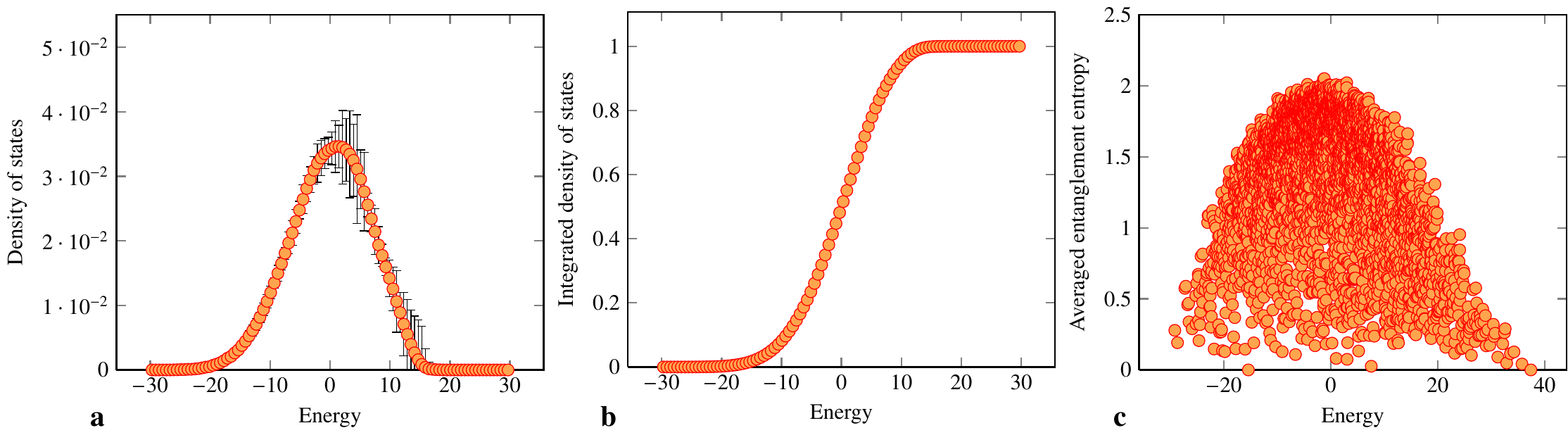}
  \caption[Density of states if Gaussian]{
    We look at a one-dimensional Heisenberg model with random on-site magnetic field of the form
    $H = \sum_{i} (X_i X_{i+1} + Y_i Y_{i+1} + Z_i Z_{i+1} + \mu_i Z_i)$,
    where $\mu_i$ is drawn uniformly from the interval $[-h,h]$.
    \textbf{a}
    Shown is the number of eigenstates for 14 sites as a function of the discretised energy for $h=1$
    averaged over 100 realisations and the variance of this averaging is indicated by vertical error bars.
    The resulting function is close to a Gaussian.
    \textbf{b}
    Here, the integrated density of states of this model with the same parameters is shown. The
    corresponding error bars are too small to appear. This quantity is related to the number of states $\Theta$
    by an exponential factor $2^N$. This plot indicates that indeed the number of states will have
    small tails and thus few states at low energies.
    \textbf{c}
    Depicted is the (von Neumann) entanglement entropy of the eigenvectors as a function of the energy.
    The results were obtained for the disorder strength $h=4$
    on a system of 12 sites and for each individual
    state an average over the different cuts through the 1D chain was taken.
    This plot corroborates the intuition that the entanglement entropy and thus the associated bond
    dimension of the corresponding MPS increase with the energy.
  }
  \label{fig:gaussian}
\end{figure*}

Part \textbf{b} of the theorem only requires Lieb-Robinson bounds on a subspace
and is thus perfectly compatible with a growth of entanglement entropies following quenches.
In fact, for the proof it is sufficient to assume Eq. \eqref{def:subsp_loc} on the level of the
individual eigenstates (see Appendix \ref{proof:part_b}).
It further allows to freely choose an energy cut-off $\kappa$ that
serves as an artificial energy gap and can be used to optimise the bound.
For typical local models, we expect the density of states to behave like a Gaussian,
an intuition that can be numerically tested for small systems (see Fig.\ \ref{fig:gaussian}),
and rigorously proved in a weak sense \cite{1403.1121}.
In this case the number
of states will behave like a low-order polynomial in the system size at energies
close to the ground state energy $E_k + \kappa$.
Thus fixing a $\kappa>0$ independent of the system size will lead to a pre-factor that scales like a low-order
polynomial in the system size where the order of the polynomial will increase as one moves to higher energies.
In the bulk of the spectrum, the number of states will grow exponentially with the system size,
thus rendering our bounds useless as one moves to high energies.
Interestingly, this feature of stronger correlations, associated with a larger entanglement in the state, at
higher energies seems to be shared by the Heisenberg chain with random on-site magnetic field (see Fig.\ \ref{fig:gaussian}).
Naturally our theorem can also be applied to the ground state where it
extends previous results \cite{HamzaSimsStolz12} and for example for
the case of an almost degenerate ground state will be a substantial improvement.
Naturally, our results can also be applied to the highly excited states at the other end of the
spectrum.

The proofs for both parts of the theorem rely on energy filtering techniques \cite{1001.5280,1102.0842}.
This is a versatile tool especially in the study of perturbation bounds for Hamiltonian systems and
allows to partly diagonalise a local observable in the energy eigenbasis, while still keeping some locality structure.
Energy filtering of an observable is defined as
\begin{align}
  I_f^H (A) = \int_{-\infty}^\infty dt f(t) A(t) .
\end{align}
Here $f:\rr\rightarrow \rr^+$ is a \emph{filter function} that can be used to interpolate between locality of the resulting
observable and the strength of the off-diagonal elements in the Hamiltonian basis
(see Appendix \ref{appendix:filtering}).

In order to show part \textbf{a} of the theorem, we will make use
of energy filterings of the local observables $A$ and $B$ with respect to the full Hamiltonian $H$
as well as with respect to local restrictions $H_A$ and $H_B$.
Those restrictions contain the support of the corresponding observable and are chosen as large as possible,
while still satisfying $[H_A, H_B]=0$.
The main idea is to choose a suitable (Gaussian) filter function $I_\alpha^H$
and use the gap assumptions \textbf{AI} and \textbf{AII} to show that the joint filter of the observables $AB$
decouples into energy filters for $A$ and $B$ separately. Then strong dynamical localisation can be used to
make these filters local.
Finally choosing the width of the energy filter $\alpha$ small enough, and in
particular exponentially small in the system size, allows to conclude the proof
(for details see Appendix \ref{part_a}).
The proof of part \textbf{b} of the theorem is more involved, relies on a high-pass filter
(see Appendix \ref{appendix:filtering}) and is contained in Appendix \ref{proof:part_b}.

{\it Implications on area laws and matrix-product states.}
In one dimension, the conclusions of our main theorem about the correlation behaviour of eigenvectors
can be turned into a statement about their entanglement structure.
It has been noted before that many-body localisation should be connected to eigenvectors
fulfilling an area law (see Conjecture 1 in Ref.\ \cite{Bauer}),
and eigenvectors being well approximated by matrix-product state vectors of the form
\begin{equation}
	\ket{{\mathrm{MPS}}} = \sum_{i_1,\dots, i_N=1}^D{\rm tr}(M_{i_1} M_{i_2}\dots M_{i_N}) |i_1,\dots, i_N\rangle,
\end{equation}
where $D$ is the bond dimension and $M_j\in \cc^{D\times D}$ for all $j$.
Our main theorem allows to rigorously prove this connection.
\begin{co}[Area laws and matrix-product states]
  \label{mps_theorem}
  An eigenvector $\ket{k}$ of a localising Hamiltonian can be approximated by an MPS with fidelity
  $|\braket{k}{\mathrm{MPS}}|\geq 1 - \epsilon$, where the bond dimension for a sufficiently large system 
  is given as follows.
  \begin{itemize}
    \item[\textbf{a}]
      If the Hamiltonian shows strong dynamical localisation
      and its spectrum fulfils assumptions \textbf{AI} and \textbf{AII},
      then the statement holds for all eigenvectors $\ket{k}$
      and, for some constant $C>0$, the approximation has a bond dimension 
      \begin{align}
        D = C \left ( {N}/{\epsilon} \right )^{\frac{16}{\mu \log_2 e}}.
      \end{align}
    \item[\textbf{b}]
      If the Hamiltonian has a mobility edge at energy $E_{\mathrm{mob}}$,
      and its spectrum fulfils assumptions \textbf{AI} and \textbf{AIII},
      then the statement holds for all eigenvectors below this energy $E_{\mathrm{mob}}$
      and the bond dimension is given by
      \begin{align}
        D = \mathrm{poly} \left ( \idos( E_k + \kappa), N \right ) ,
      \end{align}
      for any fixed $\kappa$ which enters in the precise form of the polynomial.
  \end{itemize}
\end{co}
The proof is a direct consequence of our main theorem together with the fact that exponential clustering 
in one dimension implies strong bounds on entanglement entropies 
between any bipartite cut of the chain \cite{Brandao_Clustering}. 
Using techniques from \cite{MPS_faithfully} this leads to an efficient MPS approximation with the above bounds on the 
bond dimension (see Appendix \ref{mps_proof} for details).

{\it Summary \& outlook.}
Despite considerable progress in understanding the effects of random potentials on quantum many-body systems,
a precise definition of the phenomenon of many-body localisation continues to be elusive.
Attempts to capture the phenomenology can largely be classified into two complementary approaches:
One of them puts characteristic properties of the eigenfunctions into the focus of attention
and asks for a lack of entanglement and a violation of the ETH.
The other one takes the suppression of transport as the basis, which seems closer to being experimentally testable.

In this work, we have established a clear link between these two approaches, by showing that dynamical localisation
implies that eigenvectors cluster exponentially. 
This result, together with the existence of an approximating MPS description in one dimension,
reminds of the definition of Ref.\ \cite{Bauer}, that defines many-body localisation
in terms of matrix-product state approximations of eigenstates.
In contrast, in our work, this feature is shown to follow from an absence of transport.
For future research, it would be interesting to further explore this connection and to
address the converse direction, namely to establish that an area-law for all eigenvectors implies that excitations
cannot travel through the system.

A different approach towards approximating the individual eigenstates with a matrix-product state
could potentially be provided by constructing meaningful local constants of motion that give a set of local quantum numbers
\cite{1407.8480}. It seems likely that tools using energy filtering will again prove useful in this context,
a prospect that we briefly discuss in Appendix \ref{constants}.
On the practical side, further numerical effort will be needed to understand the behaviour of individual models and
to fully understand the transport properties for different energy scales.
Our work could well provide a first stepping stone for further endeavours in this direction:
Showing that eigenstates are well approximated by matrix-product states implies that not only ground states,
but in fact also excited states can efficiently be described in terms of tensor networks.
Our result as such does not yet provide an efficient algorithm to
find the respective matrix-product states: this reminds of the situation of the existence of lattice models for which the ground states
are exact matrix-product states, but it amounts to a computationally difficult problem to find them \cite{SchuchNP}. Still, this appears to be
a major step in the direction
of formulating such numerical prescriptions of describing the low-energy sector of many-body localizing systems.
Eventually,
the leading vision in any of these endeavours appears to be a rigorous proof
of many-body localisation in the spirit of the original results by
Anderson. For this, creating a unifying framework and linking the possible definitions seems a key first step.

\emph{Acknowledgements.}
We thank H.\ Wilming, C.\ Gogolin, T.\ Osborne, F.\ Verstraete and F.\ Pollmann for insightful discussions and M.\ Goihl for collaboration on
the numerical code.
We would like to thank the EU (SIQS, AQuS, RAQUEL, COST),
the ERC (TAQ), the BMBF, and the Studenstiftung des deutschen Volkes for support.
VBS is supported by the Swiss National Science Foundation through the National Centre of Competence in Research `Quantum Science and Technology' and by an ETH postdoctoral fellowship.
WB is supported by EPSRC.


\section*{Appendix}

\subsection{Lieb-Robinson bounds}
  \label{LR}
  Zero velocity Lieb-Robinson bounds are a way to rigorously capture the suppression of transport.
  Since there is no comprehensive discussion of such bounds available in the literature, we
  precisely discuss the relationships between different possible definitions here.
  Depending on the precise notion of transport adapted, the formulation of the bound slightly varies.
  The strongest variant is that information cannot be send through a system by locally changing the
  Hamiltonian and is rigorously captured by our definition of strong dynamical localisation
  (Def.\ \ref{df_dynamical_localisation}).
  \begin{df}(Zero velocity Lieb-Robinson bounds: Truncated Hamiltonian)
    \label{LR_truncation}
    A Hamiltonian is said to satisfy zero velocity Lieb-Robinson bounds in the truncated Hamiltonian
    formulation, iff
    \begin{align}
      \label{LR_truncH}
      \| A(t) - e^{i t H_l} A e^{-i t H_l} \| & \leq \min(t,1) C e^{- \mu l}.
    \end{align}
    Here $A(t)$ is an observable in the Heisenberg picture, $H^l$ the truncated Hamiltonian
    including all Hamiltonian terms within a distance $l$ of the support of $A$ and $C,\mu>0$ are constants.
  \end{df}
  As in this whole work, $\|A\|=1=\|B\|$ was chosen for convenience.
  In constrast to the definition in the main text (Def.\ \ref{df_dynamical_localisation}),
  we here include a factor $\min(t,1)$, which is
  stronger than what we need for our proof of Theorem \ref{main_theorem} part \textbf{a}.
  However, since it seems reasonable to assume that the left hand side of Eq.\ \eqref{LR_truncH}
  initially grows continuously with $t$ and this way of stating the bound allows for an easier connection
  between the different Lieb-Robinson versions, we decided to include the factor here.
  Indeed the strongest standard Lieb-Robinson bounds in nearest-neighbour systems also include such a factor
  \cite{1308.2882}.

  Another possible and very physical way of stating that a system has no transport, is to say that excitations
  cannot spread through the system. Mathematically, this can be captured by demanding that the action of
  a local unitary cannot be detected far away, even for long times.
  \begin{df}(Zero velocity Lieb-Robinson bounds: Excitations)
    \label{LR_excitations}
    A Hamiltonian is said to satisfy zero velocity Lieb-Robinson bounds in the excitation
    formulation, iff for any locally supported Hermitian operator $G$ with $\Vert G\Vert = 1$
    \begin{align}
      &|\bra{\psi} A(t) \ket{\psi} - \bra{\psi} e^{i s G} A(t) e^{- i s G} \ket{\psi}| \\
      \leq& s \min(t,1) C^{'} e^{-\mu d(A,G)},
    \end{align}
    where $\ket{\psi}$ is an arbitrary state vector and  $C^{'},\mu>0$ are $t$ independent constants.
  \end{df}

  Most useful for mathematical manipulations is often the formulation in terms of a commutator
  \cite{1001.5280,HamzaSimsStolz12}.
  \begin{df}(Zero velocity Lieb-Robinson bounds: Commutator)
    \label{LR_commutator}
    A Hamiltonian is said to satisfy zero velocity Lieb-Robinson bounds in the commutator
    formulation, iff
    \begin{align}
      \| [A(t) , B] \| & \leq \min(t,1) C^{''} e^{- \mu \dist},
    \end{align}
    with suitable $t$ independent constants $C'',\mu>0$.
  \end{df}

  Interestingly, these definitions can be closely related, similar to the case of standard Lieb-Robinson
  bounds \cite{Barthel,OpenReview,NachtergaeleSims10}.
  \begin{lm}(Connecting zero velocity Lieb-Robinson bounds)
    The different formulations of zero velocity Lieb-Robinson bounds are related in the following way
    \begin{align}
      \text{Truncation (Def.\ \ref{LR_truncation})} &\Rightarrow \text{Commutator (Def.\ \ref{LR_commutator})}\\
      \text{Excitations (Def.\ \ref{LR_excitations})} &\Leftrightarrow \text{Commutator (Def.\ \ref{LR_commutator})},
    \end{align}
    with suitable constants $C, C^{'}, C^{''}$.
  \end{lm}
  \begin{proof}
    Going from the truncation formulation (Def.\ \ref{LR_truncation}) to the commutator formulation
    (Def.\ \ref{LR_commutator}) is straightforward
    \begin{align}
      \| [A(t) , B] \| & \leq \| [ A(t)-e^{i t H_l} A e^{- i t H_l} , B ] \| \\
                       &+ \| [e^{i t H_l} A e^{- i t H_l} , B ] \|\\
                       &\leq \min(t,1) C^{''} e^{- \mu \dist},
    \end{align}
    where $ l = d(A,B)$ has been chosen to make the second commutator zero.

    To show the one implication of the second part of the lemma, going from Def.\ \ref{LR_excitations}
    to Def.\ \ref{LR_commutator}, we rewrite the excitations formulation as
    \begin{align}
     \left|\bra{\psi} \sum_{k>0} \frac{s^k i^k}{k!} \tau_G^k (A(t)) \ket{\psi}\right|
      \leq s \min(t,1) C^{'} e^{- \mu d(A,G) } ,
    \end{align}
    where $\tau_G(A) = [G,A]$ is the commutator and $\tau^k_G$ denotes multiple applications of it.
    Dividing by $s$ and taking the limit $s\rightarrow 0$ yields
    \begin{align}
      |\bra{\psi} [A(t),G] \ket{\psi}|
      \leq \min(t,1) C^{'} e^{- \mu d(A,G)} \;.
    \end{align}
    Since the bound holds for arbitrary Hermitian observables $G$ and arbitrary state vectors $\ket{\psi}$,
    this concludes the proof.

    To show the converse direction, going from Def.\ \ref{LR_commutator} to Def.\ \ref{LR_excitations},
    we express the left hand side of Def.\ \ref{LR_excitations} as an integral
    \begin{align}
      &|\bra{\psi} A(t) \ket{\psi} - \bra{\psi} e^{i s G} A(t) e^{- i s G} \ket{\psi}|\\
      &= | \bra{\psi} \int_0^s \intd r e^{i r G} [G, A(t)] e^{- i r G} \ket{\psi}|\\
      &\leq \int_0^s \intd r \| [G, A(t) ] \|\\
      &\leq s \min(t,1) C^{''} e^{- \mu \dist},
    \end{align}
    where we used Def.\ \ref{LR_commutator} in the final step.
  \end{proof}
  Using common tools \cite{NachtergaeleSims10}, one can also use the commutator formulation to arrive at
  estimates similar to the truncation formulation, but with a slightly weaker pre-factor,
  depending on the time $t$.

\subsection{Energy filtering}
  \label{appendix:filtering}
  A tool that plays an important role in this work is energy-filtering with respect to a suitable filter function.
  Energy-filtering of a local observable is defined as follows
  \begin{align}
    I_f^H (A) = \int_{-\infty}^\infty dt f(t) A(t) \;.
  \end{align}
  Here $f:\rr\rightarrow \rr^+$ is a so-called filter function, usually taken as a $C^\infty$-function.
  Here, $A(t) = e^{it H} A e^{-itH}$ refers to time evolution under the full Hamiltonian $H$,
  but we will later also consider filters with respect to truncated Hamiltonians, such as $I_f^{H_A}$.
  Energy-filtering allows to alter the matrix elements of a local observable
  in the eigenbasis of a Hamiltonian, while still keeping some form of locality \cite{1001.5280}.
  We will work with two types of filter functions. Gaussian filter functions in particular
  provide a good compromise between locality in Fourier space and decay behaviour in real time
  and hence allow us to pick out narrow energy windows, while still preserving the approximate
  locality of the observable.

  \begin{dflm}[Gaussian filters]
    \label{filter:gaussian}
    A Gaussian filter is defined as
    \begin{align}
      I_\alpha^H(A) &:= \sqrt{\frac{\alpha}{\pi}} \int_{-\infty}^{\infty} \intd t
      e^{- \alpha t^2} A(t),
    \end{align}
    where $\alpha>0$ defines the sharpness of the filter.
    The matrix elements in the eigenbasis of the Hamiltonian fulfil
    \begin{align}
      \bra{r} I_f^H(A) \ket{s} &= \bra{r} A \ket{s} e^{-{(E_s - E_r)^2}/{(4 \alpha)}}.
    \end{align}
    For strongly localizing systems (Def.\ \ref{df_dynamical_localisation}),
    local observables filtered with a Gaussian filter still remain approximately
    local in the sense that
    \begin{align}
      \|I_\alpha^H (A)  - I_\alpha^{H_l} (A)\| \leq c_\mathrm{loc} e^{-l}.
    \end{align}
    For systems with a mobility edge (Def.\ \ref{def:subsp_loc}), we have
    \begin{align}
      \Abs{\bra{k} [I_f(A),B] \ket{k}} \leq c_\mathrm{mob} e^{-d(A,B)}.
    \end{align}
  \end{dflm}
  \begin{proof}
    The Gaussian suppression of off-diagonal elements readily follows from the definition
    \begin{align}
      \bra{r} I_\alpha(A) \ket{s} &:= \sqrt{\frac{\alpha}{\pi}} \int \intd t e^{- \alpha t^2} \bra{r} A(t) \ket{s}\\
                                  &= \bra{r} A \ket{s} \sqrt{\frac{\alpha}{\pi}} \int \intd t e^{- \alpha t^2}
        e^{i t (E_r - E_s)} ,
    \end{align}
    and the fact that the Fourier transform of a Gaussian is again a Gaussian.
    Assuming strong dynamical localisation, deriving locality is straightforward,
    \begin{align}
      &\| I_\alpha^H (A) - I_\alpha^{H_A} (A) \| \\
      =&\left \| \sqrt{\frac{\alpha}{\pi}} \int_{-\infty}^\infty \intd t e^{-\alpha t^2}
      \left (A(t) - e^{i t H_A} A e^{-i t H_A} \right) \right \|\\
      \leq& c_{\text{loc}} e^{-\mu l} ,
    \end{align}
    where have used that the Gaussian filter is normalised.
    The case with a mobility edge can be shown in the same way.
  \end{proof}

  \begin{figure}[t]
    \includegraphics[width=\columnwidth]{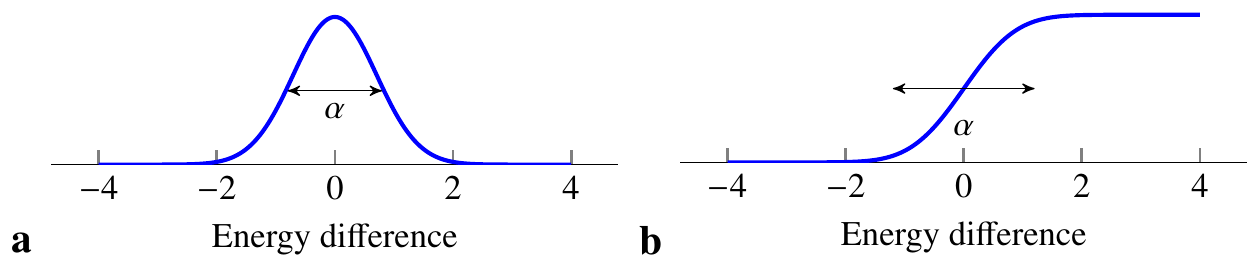}
    \caption{
      Presented is a schematic sketch of the energy filtering under a Gaussian (\textbf{a}) and a high-pass (\textbf{b})
      filter and the role of the sharpness $\alpha$.
    }
    \label{filter_sketch}
  \end{figure}

  Another important filter is the high-pass filter, which blocks large negative energy differences.
  \begin{dflm}[High-pass filters]
    \label{filter:high-pass}
    A high-pass filter is defined by
    \begin{align}
      \Gamma_{\alpha}(A) &:= \lim_{\varepsilon\rightarrow 0}
          \frac{i}{2\pi}  \int_{-\infty}^{\infty} \intd t  \frac{e^{-\alpha t^2}}{t+ i\varepsilon} A(t) .
    \end{align}
    Here $\alpha>0$ describes the sharpness of the filter. The matrix elements can be
    bounded for any $|E_s-E_r| \geq \sigma \geq0$ by
    \begin{align}\label{eq:subsp_cases}
      \frac{\bra{r} \Gamma_{\alpha}(A) \ket{s}}{\bra{r} A \ket{s}} =
      \begin{cases}
        \frac{1}{2} e^{-{\sigma^2}/{(4\alpha)}} & \text{for }\; E_r\geq E_s \\
        1 - \frac{1}{2} e^{-{\sigma^2}/{(4\alpha)}} & \text{for }\; E_r\leq E_s
      \end{cases}\;.
    \end{align}
    Local observables remain approximately local under a high-pass filter, in the sense that
    \begin{align}
      |\bra{k} [\Gamma_{\alpha}(A),B] \ket{k}| \leq
      \frac{e^{-\mu \dist}}{2 \pi}
      \left(4 + \ln \frac{\pi}{4 \alpha} \right) .
    \end{align}
  \end{dflm}
  \begin{proof}
    Calculating the off-diagonal elements of such a high-pass filter relies
    on the subsequent bound on the error function of a
    Gaussian random variable proven in Ref.\ \cite{HastingsKoma06}, stated here as 
    Lemma \ref{lem:gauss_tails}.
    Making use of the identity proven in Ref.\ \cite{HastingsKoma06}
    \begin{align}
      \lim_{\varepsilon\rightarrow 0} &\frac{i}{2\pi} \int_{-\infty}^\infty
      \frac{e^{-\alpha t^2} e^{-i\Delta E \, t}}{t+ i\varepsilon} \intd t\\
      = &\frac{1}{2\pi} \sqrt{\frac \pi \alpha} \int_{-\infty}^0
      e^{-{\left(\omega+\Delta E\right)^2}/{(4\alpha)}}\intd\omega ,
    \end{align}
    the bounds for the matrix elements of $\Gamma_{\alpha}(A)$ follow from Lemma \ref{lem:gauss_tails}.
    The locality statement can be shown by splitting the integral into three parts \cite{HamzaSimsStolz12}
    and using our assumption of dynamical localisation
    \begin{multline}
      \bra{k} [\Gamma_{\alpha}(A),B] \ket{k}
      \leq \frac{e^{-\mu \dist}}{2\pi}\times\\
      \left(\int_{\Abs{t}\leq 1} e^{-\alpha t^2} \intd t  + \int_{1\leq\Abs{t}\leq \lambda}
      \frac{e^{-\alpha t^2}}{t} \intd t + \int_{\Abs{t}\geq\lambda} \frac{e^{-\alpha t^2}}{t} \intd t\right)\\
      \leq \frac{e^{-\mu \dist}}{\pi}
      \left(1 + \ln\lambda + \frac{1}{2\lambda} \sqrt{\frac{\pi}{\alpha}} \right) .
    \end{multline}
    Here estimating the first term used the $\min(1,t)$ factor included in the definition of a mobility edge
    (Def. \ref{def:subsp_loc}).
    Choosing
    \begin{equation}
    \lambda = \frac{\sqrt{\pi}}{2 \sqrt{\alpha}}
   \end{equation}
   concludes the proof.
  \end{proof}

   \begin{lm}[Hastings-Koma]\label{lem:gauss_tails}
      Let $E\in \rr$, $\alpha>0$ then for all $\gamma>0$ with $E\leq-\gamma$
      \begin{align}
        \frac{1}{2\pi} \sqrt{\frac \pi \alpha} \int_{-\infty}^0
        e^{-\frac{\left(\omega+E\right)^2}{4\alpha}}\intd\omega
        \leq \frac{1}{2} e^{-{\gamma^2}/({4\alpha})}
      \end{align}
      and for all $\gamma>0$ with $E \geq \gamma$
      \begin{align}
        \left | \frac{1}{2\pi} \sqrt{\frac \pi \alpha} \int_{-\infty}^0
        e^{-\frac{\left(\omega+E\right)^2}{4\alpha}}\intd\omega  - 1 \right |
        \leq \frac{1}{2} e^{-{\gamma^2}/({4\alpha})} .
      \end{align}
    \end{lm}

  For later use, we give the Gaussian and high pass filters of sharpness $\alpha$ the
  symbols $I_\alpha^H$ and $\Gamma^H_\alpha$, respectively.
  In case the local observable is filtered with the full system Hamiltonian,
  we will often omit the $H$.
  It is an interesting insight that the locality structure of a Gaussian filter is independent of its sharpness
  $\alpha$ for dynamically localising systems, while it still depends on $\alpha$ for a high-pass filter.
  It is precisely this property of Gaussian filters that will allow us to prove part $\textbf{b}$
  of Theorem \ref{main_theorem}.

\subsection{Local constants of motion}
  \label{constants}
  In this appendix, we briefly comment on how local constants of motion could be constructed using
  energy filtering.
  A natural candidate for local constants of motion
  are the energy-filtered versions of the local Hamiltonian terms
  \begin{align}
    M_j := I_\alpha^H (h_j) .
  \end{align}
  Since they act non-trivially on most eigenstates, it is clear that
  they are retained by the energy-filtering without being averaged to zero and thus indeed provide useful quantum numbers.
  What is more, strong dynamical localisation implies that they will stay approximately local.
  This construction can be optimised by choosing suitable filter functions, allowing for an interpolation
  between the locality of the observables and making them completely constant in time.
  In particular, taking the time-average is equivalent to choosing a constant filter function \cite{1407.8480}
  leading to quantities that are completely preserved in time, while they will usually not be strictly local anymore.
  Naturally, this approach is not limited to taking the local Hamiltonian terms and in principle any local
  observable that does not vanish after applying the energy filter will provide useful quantum numbers.
  Specifically, polynomials of the local Hamiltonian terms can equally well be used.
  In systems with a mobility edge, where not all transport is suppressed, this construction is still possible,
  but the resulting constants of motion might no longer be approximately local.

  Since local constants of motion by definition commute with the Hamiltonian, their local eigenbasis
  is compatible with the global energy eigenbasis of the Hamiltonian.
  Thus, they could be used to directly construct local states in a matrix-product state form, as
  long as enough constants of motion exist in order to use their quantum numbers to divide the full Hilbert
  space into suitable small fractions.

\subsection{Proof of Theorem \ref{main_theorem} part a}
  \label{part_a}
  In this appendix, we will formulate the details of the proof of Theorem \ref{main_theorem} part \textbf{a},
  which shows that strong dynamical localisation implies exponential clustering of all eigenvectors.
  For this, we will rely on the following lemma.

  \begin{figure}
    \includegraphics[]{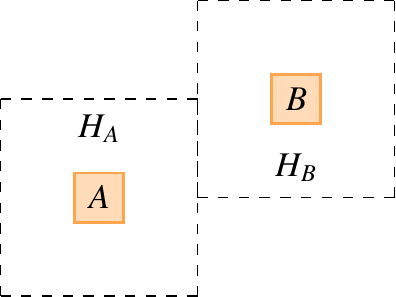}
    \caption{
      This figure shows the Hamiltonian decomposition used in the proof for part \textbf{a} of our main theorem.
    }
    \label{fig:Hamiltonian_decomposition}
  \end{figure}

  \begin{lm}[Decoupled energy filtering]
  \label{lm_energy_nondegenerate}
    Under the assumption of locally independent gaps (\textbf{AII}),
    energy-filtering of two observables can be factorized into local energy filters
    \begin{align}
      \| I_\alpha^{H_A + H_B} (A B) - I_\alpha^{H_A} (A) I_\alpha^{H_B} (B) \|
      \leq 2^{4 N +1} e^{-{\xi^2}/({4 \alpha})} ,
    \end{align}
    with $\xi=\min \{\eta,\sqrt{2} \tilde{\gamma}\}$.
    Here $H_A$ and $H_B$ are chosen to include all Hamiltonian terms within distance
    $d(A,B)/2$ of the support of $A$ and $B$ respectively (See Fig. \ref{fig:Hamiltonian_decomposition}).
  \end{lm}
  \begin{proof}
    To show that instead of applying an energy filter to $AB$, we can also
    apply it to the observables individually, we need to use that
    the eigenvalues of $H_A + H_B$ are disconnected on the two regions
    and we can thus label them by two different quantum numbers $a,b$.
    We hence get
    \begin{align}
     & \biggl \| I_\alpha^{H_A + H_B} (A B)
      - \sum_{a,b} \ketbra{a,b}{a,b} AB \ketbra{a,b}{a,b}\\
      &+ \sum_{a,b} \ketbra{a,b}{a,b} AB \ketbra{a,b}{a,b}
      - I_\alpha^{H_A} (A) I_\alpha^{H_B} (B) \biggr\|
      .
    \end{align}
    We will use the triangle inequality and proceed to show that both energy
    filters give only the diagonal entires up to a small error.
    Starting with the first term gives the following estimate,
    \begin{align}
      &\|I_\alpha^{H_A} (A) I_\alpha^{H_B} (B)
      -\sum_{a,b} \ketbra{a,b}{a,b} AB \ketbra{a,b}{a,b}\|\\
      =&\biggl \|\sum_{a,b} \sum_{a' \neq a,b' \neq b}
        \ketbra{a,b}{a,b} AB \ketbra{a',b'}{a',b'}\\
        &\exp\left(-\frac{(E_a - E_{a'} + E_b - E_{b'})^2}{4 \alpha}\right)
        \biggr\|\nonumber
        \\
      =& \sum_{a,b} \sum_{a' \neq a,b' \neq b}
        \biggl\|\ketbra{a,b}{a,b} AB \ketbra{a',b'}{a',b'}\\
        &\exp\left({-\frac{(E_a - E_{a'} + E_b - E_{b'})^2}{4 \alpha}}\right)\biggr\|\nonumber\\
        =& 2^{4 N} e^{-{\eta^2}/{(4 \alpha)}} ,
    \end{align}
    where we assumed locally independent gaps (\textbf{AII}).
    The second term yields the following estimate,
    \begin{align}
      &\|I_\alpha^{H_A} (A) I_\alpha^{H_B} (B)
      -\sum_{a,b} \ketbra{a,b}{a,b} AB \ketbra{a,b}{a,b}\|\\
      =& \biggl\|\sum_{a,b} \sum_{a'\neq a,b' \neq b}
        \ketbra{a,b}{a,b} AB \ketbra{a',b'}{a',b'}\\
        &\exp\left({-\frac{(E_a - E_{a'})^2}{4 \alpha}}\right)
        e^{-\frac{(E_b - E_{b'})^2}{4 \alpha}}\biggr\|\\
      =& \biggl \|\sum_{a,b} \ketbra{a,b}{a,b} AB
        \widetilde{D}_{\tilde{\gamma}}^a \widetilde{D}_{\tilde{\gamma}}^b \biggr\|
        \leq 2^{2N} e^{-2 \tilde{\gamma}^2/(4 \alpha)} .
    \end{align}
    Here, $\widetilde{D}_{\tilde{\gamma}}^a$ is a diagonal matrix with entries
  $e^{-(E_a - E_a')^2/(4 \alpha)}$ for $a\neq a'$ and 0 otherwise
  and we assumed locally independent gaps (\textbf{AII}).
  \end{proof}

  Finally let us briefly comment that the local independence of
  gaps can alternatively be ensured by demanding that the joint Hamiltonian $H = H_A + H_B$
  supported on two disjoint rectangular regions has non-degenerate energies, since these
  energies are all possible sums of the local energies $E_a + E_b$.
  \newline

  \textbf{Proof of Theorem \ref{main_theorem}}
  We are now in the position to prove part \textbf{a} of our main theorem.
  \setcounter{thm}{0}
  \begin{thm}[Clustering of correlations of eigenvectors]
    The dynamical properties of a local Hamiltonian imply clustering correlations
    of its eigenvectors in the following way.
    \begin{itemize}
      \item[\textbf{a}]
        If the Hamiltonian shows strong dynamical localisation
        and its spectrum fulfils assumptions \textbf{AI} and \textbf{AII}
        then all its eigenvectors have exponentially clustering correlations
        \begin{align}
          |\bra{k} A B \ket{k} - \bra{k} A \ketbra{k}{k} B \ket{k}|
          \leq 4 c_\mathrm{loc} e^{-{\mu} \dist /2} .
        \end{align}
    \end{itemize}
  \end{thm}

  \begin{proof}
    The proof uses the local Hamiltonians $H_A$ and $H_B$ (See Fig. \ref{fig:Hamiltonian_decomposition})
    and the following estimates
  \begin{align}
    \bra{k} AB \ket{k}
    &= \bra{k} I_\alpha^H ( AB ) \ket{k} \\
    &\displaystyle_{I}^{=}
    \bra{k} I_\alpha^{H_A + H_B} ( AB ) \ket{k} + \Omega(e^{- \mu l})\\
    &\displaystyle_{II}^{=}
    \bra{k} I_\alpha^{H_A} (A) I_\alpha^{H_B} (B) \ket{k} + \Omega(e^{- \mu l}) \\
    &\displaystyle_{III}^{=}
    \bra{k} I_\alpha^{H} (A) I_\alpha^{H} (B) \ket{k} + \Omega(e^{- \mu l}) \\
    &\displaystyle_{IV}^{=}
    \bra{k} A \ketbra{k}{k} B \ket{k}  + \Omega(e^{- \mu l}) .
  \end{align}
  In each step, an error term of the form $C_{I} e^{- \mu l}$ with $C_I$ to $C_{IV}$ will be introduced.
  The constant $C$ appearing in the main theorem will simply be the sum of them.
  Constants $C_I,C_{III},C_{IV}$ directly follow from the properties of the applied Gaussian filter
  (Def.\ \ref{filter:gaussian})
  and constant $C_{II}$ is derived in Lemma \ref{lm_energy_nondegenerate}
  \begin{align}
    C_{I} &= 2 c_\mathrm{loc},\\
    C_{II} &= 2^{4 N +1} e^{-{\xi^2}/({4 \alpha})} e^{\mu l},\\
    C_{III} &= c_\mathrm{loc},\\
    C_{IV} &=  e^{{- \gamma^2}/({4 \alpha})} e^{\mu l},\\
    C :&= C_I + C_{II} + C_{III} + C_{IV} .
  \end{align}
  The constants $C_{II}$ and $C_{IV}$ can be chosen arbitrarily small, by picking a sharp enough filter
  function, meaning a sufficiently small $\alpha$.
  In particular, we can choose $C_{II},C_{IV} \leq  c_\mathrm{loc}/2$, which yields
  \begin{align}
    C = 4 c_\mathrm{loc} .
  \end{align}
  Choosing $l = \dist/2$ concludes the proof.
  \end{proof}

\subsection{Proof of Theorem \ref{main_theorem} part \textbf{b}}
  \label{proof:part_b}
  In this appendix, we will show the second part of Theorem \ref{main_theorem} based on a mobility edge.
  For convenience, we start by repeating Definition \ref{def:subsp_loc} of a mobility gap.
  \setcounter{dfrestated}{1}
  \begin{dfrestated} [Mobility edge]
    A Hamiltonian $H$ is said to have a mobility edge at energy $E_{\mathrm{mob}}$
    iff its time evolution satisfies for all times $t$,
    $\forall \rho\in\vspan\{\ketbra{l}{k}: \; E_l,E_k \leq E_{\mathrm{mob}}\}$
    \begin{align}
      | \mathrm{tr} \left(\rho [A(t),B]\right)\vert
      \leq \min(t,1) c_\mathrm{mob} e^{-\mu \dist}\; ,
    \end{align}
    where $c_\mathrm{mob}$ is a $t$ independent constant.
    That is, all transport is suppressed for states supported only on the low-energy sector
    below energy $E_{\mathrm{mob}}$.
  \end{dfrestated}

  The first step for our proof will be to show that each eigenvectors only contributes to the
  correlation function by a term that is exponentially suppressed with the distance of the two
  observables.
  For this, we need the localisation assumption only for the eigenvector $\ket{k}$.
  \begin{lm}\label{lem:exp_dec_matrix_elem}
    Let $A,B$ be local observables and $\ket{k}$ a weakly localised eigenvector of $H$, i.e. satisfying
    \begin{align}
      \sup_{t\in [0,\infty)}\Abs{\bra{k}[A(t),B] \ket{k}}\leq c_{\rm mob} e^{-\mu \dist} .
    \end{align}
    Then the contribution of any eigenvector $\ket{l}, l \neq k$,
    to the correlation function will be exponentially suppressed, in the sense that
    \begin{align}
      \Abs{\bra{k} A \ketbra{l}{l} B \ket{k}}
      \leq 2 c_\mathrm{mob} e^{-\mu \dist}.
    \end{align}
  \end{lm}
  \begin{proof}
    The proof of this lemma again uses an energy filter with a Gaussian filter function (Def.\ \ref{filter:gaussian})
    and works with $\bra{k} B \ket{k} = 0$, which can always be achieved by using a shifted observable
    $\tilde{B} = B - \bra{k} B \ket{k}$.
    For this proof the filter function will be multiplied by a complex factor $e^{i t (E_k - E_l)}$
    such that its Fourier transform approximately suppresses all transitions except the one from level $k$ to level $l$.
    In a mild variant of the above filter function, we define $f:\rr\rightarrow \rr^+$ and the corresponding filter as
    \begin{align}
      f(t) &:= \sqrt{\frac{\alpha}{\pi}} e^{i t (E_k - E_l)} e^{-\alpha t^2},\\
      I_f(A) &:= \int \intd t f(t)  A(t) ,\\
      \bra{r} I_f(A) \ket{s} &= \bra{r} A \ket{s} e^{-{((E_s - E_r) - (E_l - E_k))^2}/({4 \alpha})}.
    \end{align}
    With this energy filter, we can proceed to prove the lemma.
    \begin{align}
      &\Abs{\bra{k} A \ketbra{l}{l} B \ket{k}} \\
      = &\Abs{\bra{k} I_f(A) B \ket{k} - \sum_{m \neq l} \bra{k} A \ketbra{m}{m} B \ket{k} e^{- {(E_m - E_l)^2}/({4 \alpha})}}\\
      \leq & \Abs{\bra{k} [I_f (A),B] \ket{k}} + \Abs{\bra{k} B I_f(A) \ket{k}} + 2^N e^{-{\gamma^2}/({4 \alpha})}\\
      \leq & c_\mathrm{mob} e^{-d(A,B)} + 2^N \left( e^{-{\gamma^2}/({4 \alpha})} + e^{-{\zeta^2}/({4 \alpha})} \right).
    \end{align}
    Here, $\gamma>0$ refers to the smallest gap and $\zeta>0$ is the smallest degeneracy of the gaps for fixed $E_k$
    as defined in the main text in \textbf{AI} and \textbf{AIII}
    and in the last step we used weak dynamical localisation of the eigenvector $\ket{k}$.
    We can now conclude the proof by choosing $\alpha$ small enough such that
    \begin{align}
      2^N \left( e^{-{\gamma^2}/{(4 \alpha)}} + e^{-{\zeta^2}/{(4 \alpha)}} \right ) \leq c_\mathrm{mob} e^{- \mu \dist},
    \end{align}
    which means that we pick a filter function that is narrow enough in Fourier space.
    Setting $\xi := \min(\gamma, \zeta)$ it would even be enough to choose
    \begin{equation}
    	\frac{\xi^2}{4}\left(\ln \frac{2^{2N +1}}{c_{\rm mob}}\right)^{-1}\geq \alpha,
    \end{equation}
    independently of $\dist$.
  \end{proof}

  With this, we are ready to prove part \textbf{b} of Theorem \ref{main_theorem}.
  \setcounter{thm}{0}
  \begin{thm}[Clustering of correlations of eigenvectors]
    The dynamical properties of a local Hamiltonian imply clustering correlations
    of its eigenvectors in the following way.
    \begin{itemize}
      \item[\textbf{b}]
        If the Hamiltonian has a mobility edge at energy $E_\mathrm{mob}$
        and its spectrum fulfils assumptions \textbf{AI} and \textbf{AIII},
        then all eigenvectors up to that energy $E_\mathrm{mob}$ cluster exponentially
        \begin{align}
          \nonumber
          &\Abs{\bra{k} A B \ket{k} - \bra{k} A\ketbra{k}{k} B \ket{k}}\\
          \leq&
          \left(12 \pi \idos(E_k + \kappa) c_\mathrm{mob} + \ln \frac{\pi \mu \dist e^{4 + 2 \pi}}{\kappa^2}
          \right) \frac{e^{-\mu \dist}}{2 \pi} \;,
        \end{align}
        where $\idos(E)$ is the number of states at energy $E$ and $\kappa$ is a constant
        that can be chosen arbitrarily to optimise the bound.
    \end{itemize}
  \end{thm}

  \begin{proof}[Proof of Theorem \ref{main_theorem} part \textbf{b}]
    The proof runs along and builds upon the lines of thought of both Ref.\ \cite{HastingsKoma06}
    and of \cite{HamzaSimsStolz12}, and generalises both.
    The basic idea is again to start from the correlation function
    and to transform it into an expression depending on the commutator.
    For this, we fix a constant $\kappa>0$ and define
    \begin{align}
      P=\sum_{E_l\leq E_k +\kappa}\ketbra{l}{l}
    \end{align}
    to be the projector onto the subspace of energies smaller than $E_k + \kappa$ and write $P^\perp$ for $\idty-P$.
    Decomposing the identity with respect to $P$ and $P^\perp$ and using a high-pass filter
    $\Gamma_{\alpha}$ for a suitable $\alpha>0$ (Def.\ \ref{filter:high-pass}),
    we separate the correlator into terms
    \begin{eqnarray} \label{eq:subsp_decomp}
      \bra{k} A B \ket{k} &=& \bra{k} (A-\Gamma_{\alpha}(A))P B \ket{k}\\
      &+& \bra{k} (A-\Gamma_{\alpha}(A)) P^\perp B \ket{k}\\
      &+& \bra{k} B P \Gamma_{\alpha}(A) \ket{k}
      + \bra{k} BP^\perp\Gamma_{\alpha}(A) \ket{k}\\
      &+& \bra{k} [\Gamma_{\alpha}(A),B] \ket{k}\;.
    \end{eqnarray}
    Lets consider the first term on the right-hand side of Eq.\ \eqref{eq:subsp_decomp} that contains $P$
    and expand the projector in the eigenbasis of $H$, which gives
    \begin{align}
      &|\bra{k} (A-\Gamma_{\alpha}(A))P B \ket{k}|\\
      \leq& \sum_{E_l\leq E_k+\kappa} |\bra{k} (A-\Gamma_{\alpha}(A)) \ketbra{l}{l} B \ket{k}|\\
      \leq& \; 2 \sum_{E_l\leq E_k+\kappa} |\bra{k} A \ketbra{l}{l} B \ket{k}|.
    \end{align}
    Here, we have used that all matrix elements decrease under a high-pass filter (Lemma \ref{filter:high-pass})
    for all $\alpha>0$.
    As before, we will set $\bra{k} B \ket{k} = 0$ w.l.o.g.\
    and then proceed by bounding all the terms in the sum individually.
    Using Lemma \ref{lem:exp_dec_matrix_elem} yields
    \begin{align}
      |\bra{k} (A-\Gamma_{\alpha}(A))P B \ket{k}| \leq 4 \idos(E_k + \kappa) c_\mathrm{mob} e^{- \mu \dist},
    \end{align}
    where $\idos(E_k + \kappa)$ simply is the number of eigenvectors contained in $P$.
    The other term in Eq.\ (\ref{eq:subsp_decomp}) containing $P$ can be bounded analogously, yielding
    \begin{align}
      |\bra{k} B P \Gamma_{\alpha}(A) \ket{k}| \leq 2 \idos(E_k + \kappa) c_\mathrm{mob} e^{- \mu \dist}.
    \end{align}
    The terms containing $P^\perp$ in Eq.\ (\ref{eq:subsp_decomp}) are bounded using the explicit form
    of the matrix elements of a high-pass filter (Lemma \ref{filter:high-pass}), which delivers
    \begin{align}
      \bra{k} B P^\perp \Gamma_{\alpha}(A) \ket{k}
      \leq \sum_{E_l > E_k+\kappa} \bra{k} B \ketbra{m}{m} D A \ket{k},
    \end{align}
    where $D$ is a diagonal matrix whose entries are bounded by $e^{-{\kappa^2}/{(4 \alpha)}}/2$.
    A simple norm estimate concludes the estimation of this term
    \begin{align}
      \bra{k} B P^\perp \Gamma_{\alpha}(A) \ket{k}  \leq \|D\|
      \leq \frac{e^{-{\kappa^2}/{(4 \alpha)}}}{2} .
    \end{align}
    The other term containing $P^\perp$ can be estimated in the same way
    \begin{align}
    \bra{k} \left( A - \Gamma_{\alpha}(A) \right) P^\perp \Gamma_{\alpha}(A) \ket{k}
      \leq \frac{e^{-{\kappa^2}/{(4 \alpha)}}}{2} .
    \end{align}
    The commutator term in Eq.\ (\ref{eq:subsp_decomp}) is bounded by using locality of a high-pass filter
    (Lemma \ref{filter:high-pass}). We are still free to choose a value for $\alpha$ in the high-pass filter.
    Taking
    \begin{equation}
    	\alpha = \frac{\kappa^2}{4 \mu \dist}
    \end{equation}
    gives
    \begin{eqnarray}
      \bra{k} [\Gamma_{\alpha}(A),B] \ket{k} &\leq&
      \bigg( 12 \pi \idos(E_k + \kappa) c_\mathrm{mob} \\
        &+& 2 \pi + 4 + \ln \frac{\pi \mu \dist}{\kappa^2}
      \bigg) \frac{e^{-\mu \dist}}{2 \pi} \;.
    \end{eqnarray}
    This concludes the proof.
  \end{proof}

\subsection{Proof of Corollary \ref{mps_theorem}}
  \label{mps_proof}
  It follows from the above results that in the case of 1D systems,
  the corresponding eigenstates satisfy an area law for a suitable
  entanglement entropy, which in turn implies that they can be
  well approximated by a matrix-product state of a low bond dimension.
  We will bound the bond dimension by using Theorem 1 from Ref.\ \cite{Brandao_Clustering}.
  For convenience, we repeat our Corollary \ref{mps_theorem} from the main text.
  \setcounter{thm}{1}
  \begin{co}[Area laws and matrix-product states]
    An eigenvector $\ket{k}$ of a localising Hamiltonian can be written as an MPS with fidelity
    $|\braket{k}{\mathrm{MPS}}|\geq 1 - \epsilon$, where the bond dimension for a sufficiently large system size
    is given as follows.
    \begin{itemize}
      \item[\textbf{a}]
        If the Hamiltonian shows strong dynamical localisation
        and its spectrum fulfils assumptions \textbf{AI} and \textbf{AII},
        then the statement holds for all eigenvectors $\ket{k}$
        and the approximation has a bond dimension
        \begin{align}
          D = C \left ( {N}/{\epsilon} \right )^{\frac{16}{\mu \log_2 e}},
        \end{align}
        for some constant $C>0$.
      \item[\textbf{b}]
        If the Hamiltonian has a mobility edge at energy $E_{\mathrm{mob}}$,
        and its spectrum fulfils assumptions \textbf{AI} and \textbf{AIII},
        then the statement holds for all eigenvectors below this energy $E_{\mathrm{mob}}$
        and the bond dimension is given by
        \begin{align}
          D = \mathrm{poly} \left ( \idos( E_k + \kappa), N \right ) \; ,
        \end{align}
        for any fixed $\kappa$ which enters in the precise form of the polynomial.
    \end{itemize}
  \end{co}
  \begin{proof}
  \begin{figure}
    \includegraphics[]{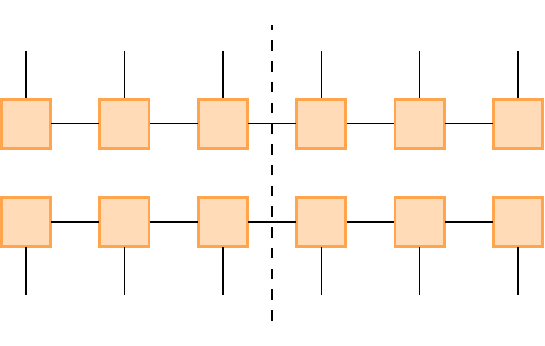}

    \vspace{-.4cm}

    \includegraphics[]{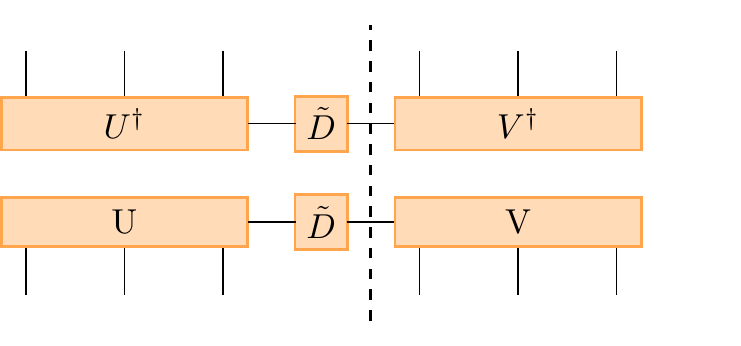}
    \caption{
      MPS description of an eigenstate $\ketbra{E_k}{E_k}$.
      Theorem 1 in Ref.\ \cite{Brandao_Clustering} allows to bound the bond dimension across cuts
      by providing an upper bound to the smooth max entropy.
      The reduced state of the suitably normalised matrix product state is $\rho = U \tilde{D}^2 U^\dagger$.
    }
    \label{fig_mps}
  \end{figure}
  As a first step of the proof, we reformulate our clustering of correlation results (Theorem \ref{main_theorem}) in
  terms of a correlation length according to Ref.\ \cite{Brandao_Clustering}.
  For this, we need a decay of the correlator of the form $2^{-\dist / \xi} \; \forall \dist > l_0$ and
  start by rewriting our bound as
  \begin{align}
    \label{eq_mps_clustering}
    |\bra{k} A B \ket{k} - \bra{k} A \ketbra{k}{k} B \ket{k}|
    &\leq C e^{-\mu \dist}\\
    &= C 2^{- \mu \log_2 e \; \dist}.
  \end{align}
  Next we split the decaying term and use one part to get rid of the constant prefactor $C$
  and the other to preserve an exponential decay with a correlation length $\xi := {2}/({\mu \log_2 e})$.
  This yields
  \begin{align}
    |\bra{k} A B \ket{k} - \bra{k} A \ketbra{k}{k} B \ket{k}| &\leq 2^{- \dist / \xi} \\
    \forall \; \dist &\geq l_0,
  \end{align}
  where $l_0 = {\xi}/({\log_2 C})$.
  Following Ref.\ \cite{Brandao_Clustering}, we will use the exponential clustering to obtain a description
  of the eigenvector $\ket{k}$ in terms of MPS (see also Fig.\ \ref{fig_mps}).
  This is based on the so called smooth max entropy, which is defined as
  \begin{align}
    H_\mathrm{max}^\delta(\rho_X) := \min_{\tilde{\rho}_X \in \mathcal{B_\delta}(\rho_X)}
    \log_2 (\mathrm{rank} (\tilde{\rho}_X)) \; ,
  \end{align}
  with
  \begin{equation}
  	B_\delta(\rho_X) := \{ \tilde{\rho}_X : D(\rho_X, \tilde{\rho}_X) < \delta \}
  \end{equation}
  being a ball around $\rho_X$ (see Appendix A of Ref.\ \cite{Brandao_Clustering} for details).
  Here $D(A,B) :=  \| A - B \|_1/2$ denotes the trace distance.
  According to Theorem 1 in Ref.\ \cite{Brandao_Clustering}, we can obtain an upper bound
  for the smooth max entropy for any bipartite cut, as long as the system size is larger than
  $N \geq  {C_\mathrm{cut} l_0}/{\xi}$ with a constant $C_\mathrm{cut}>0$.
  Picking the approximation parameter in the smooth max entropy to be $\delta(l) = 2^{- {l}/{(8 \xi)}}$,
  then Ref.\ \cite{Brandao_Clustering} provides a bound of the form
  \begin{align}
    \label{eq_Hmax}
    H_\mathrm{max}^{\delta(l)} \leq c' l_0 \xi^{c \xi} + l,
  \end{align}
  for all $l \geq 8 \xi$ with constants $c, c'>0$.
  In the following, we will make Corollary 3 in Ref.\ \cite{Brandao_Clustering} explicit, by deriving
  concrete bounds on the bond dimension of the matrix-product state chosen to approximate $\ket{k}$.
  Naturally, the first step for this is to express $\ket{k}$ as a matrix-product state vector with, a priori,
  exponentially large bond dimension (Fig.\ \ref{fig_mps}).
  For simplicity, we restrict the proof here to a linear system with open boundary conditions, but
  it can equivalently be reformulated for a periodic system on a ring \cite{Brandao_Clustering}.
  The bond dimension of the matrix-product state will be bounded by truncating at each cut explicitly.
  For this, we start at one end of the chain and look at each cut separately.
  In each step, we apply a singular value decomposition (see also Fig.\ \ref{fig_mps}). Since we are truncating spectral values
  in each step, the positive operators will no longer be normalised to unit trace and are states only up to normalisation.
  This results in a reduction on the left side of the cut of $\rho = U \tilde{D}^2 U^\dagger$, with $U$ being unitary and $\tilde D$ diagonal.
  The goal is now to truncate the diagonal matrix $\tilde{D}$ to a fixed bond dimension $D$ while creating only a
  small discarded weight \cite{Schollwock201196}.
  Following Lemma 14 in Appendix B of Ref.\ \cite{Brandao_Clustering}, for any $\nu>0$,
  we can choose the bond dimension as
  \begin{align}
    D = 2^{H_\text{max}^{\nu}(\rho)} ,
  \end{align}
  and create a discarded weight
  \begin{align}
    \sum_{D+1}^{2^N} \lambda_k \leq 3 \nu,
  \end{align}
  where $\lambda_k$ are the eigenvalues of $\rho$ or equivalently of $\tilde{D}^2$.
  Each time we create discarded weight, the fidelity with our original pure quantum
  state potentially is reduced by the discarded weight \cite{MPS_faithfully}, upper bounded by $3 \nu$.
  The truncation will result in a subnormalised state. Renormalising will, however, only increase the fidelity,
  allowing us to obtain a normalised MPS approximation.
  Thus, in total, the fidelity of our MPS approximation with bond dimension
  $D=2^{H_\text{max}^{\nu}(\rho)}$ is bounded by
  \begin{align}
    |\braket{\psi}{\text{MPS}}| \leq 1 - N 3 \nu,
  \end{align}
  where $N$ is the size of the linear 1D system.
  For a fixed global error $\epsilon>0$, we therefore obtain an allowed local error
  \begin{align}
    \nu = \frac{\epsilon}{3 N} \; .
  \end{align}
  We now fix $\nu$ to take the role of $\delta(l)$ in \eqref{eq_Hmax}.
  Plugging this bound into Eq.\ \eqref{eq_Hmax} and using $\delta(l) = 2^{- {l}/{(8 \xi)}}$, we obtain
  \begin{align}
    H_\mathrm{max}^{\frac{\epsilon}{3 N}} \leq c' l_0 \xi^{c \xi} + 8 \xi \log_2 \frac{3 N}{\epsilon} \; ,
  \end{align}
  and thus obtain an approximation with
  \begin{align}
    D = 2^{ c' l_0 \xi^{c \xi}} \left ( \frac{3 N}{\epsilon} \right )^{8 \xi}  .
  \end{align}
  Coming back to our original clustering assumption in Eq.\ \eqref{eq_mps_clustering}, this yields
  \begin{align}
    D = C^{c' \xi^{c \xi + 1}} \left ( \frac{3 N}{\epsilon} \right )^{8 \xi } ,
  \end{align}
  with $\xi = {2}/({\mu \log_2 e})$.
  Thus, we finally have a polynomial scaling in the constant $C>0$ as well as the in the system size and in the
  inverse error of the approximation ${1}/{\epsilon}$.
  The rest of the proof follows directly by inserting the constants from Theorem \ref{main_theorem}.
  \end{proof}

\subsection{Absence of thermalisation}
\label{absence}

In this subsection we make the link between dynamical localisation and the absence of thermalisation more explicit.
On intuitive grounds, such a connection is much expected: if no transport happens,
then the system will retain too much memory of the initial state for all times.
This initial state dependence, then necessarily leads to an absence of thermalisation.
Here, we will formulate the link between strong dynamical localisation in the sense of Definition
\ref{df_dynamical_localisation} and the absence of thermalisation along the lines of Ref.\ \cite{Integrable}
and consistent with Ref.\ \cite{HuseReview}.

To state this concisely, denote with $\rho$ a Gibbs state of
some inverse temperature $\beta>0$. We denote with $A$ both an observable as well as the support of $A$ on the lattice. $B$ is the
region of distance no more than $l$ from $A$.
The following lemma shows that under the condition of strong localisation, there is an initial state $\rho_0$ different from $\rho$ on $A$ only, such that
$\rho_0(t)$ will still be locally statistically distinguishable from $\rho$ for all times $t>0$ on the region $B$, in stark
contrast to a presumed thermalisation. That is to say, in the sense of Ref.\ \cite{Integrable} the system retains too much of a memory of the initial condition to thermalise.
Hence, the system must also violate the ETH.
\begin{lm}[Absence of thermalisation] Let $\rho$ be a Gibbs state. Then there exists a state $\rho_0$, different from $\rho$ on $A$ only, such that
the trace distance of the time evolved state $\rho_0(t)$ and the Gibbs state $\rho$ reduced to $B$
is lower bounded by
\begin{equation}
	\|\rho_0(t)_B - \rho_B\|_1\geq x- c_\mathrm{loc} e^{- \mu l},
\end{equation}
with $x:= 2- 2\lambda_{\rm min}(\rho_A)$, where $\rho_A$ and $\rho_B$ are the reductions to $A$ and $B$, respectively.
\end{lm}
\begin{proof}
For the distinguished region of the lattice, refer to $\rho_A$ as its reduced state and $\rho_{\backslash A}$ its complement.
Take as initial state $\rho_0=\rho_{\backslash A}\otimes \xi_A$, where $\xi_A$ is chosen such that
\begin{equation}
	\|\rho_A-\xi_A\|_1 =  2- 2\lambda_{\rm min}(\rho_A) = x.
\end{equation}
Using that the trace norm is a unitarily invariant norm, it is clear that such a state $\xi_A$ can always be found.
Denote with $A$ the observable satisfying $\|A\|\leq 1$ that saturates the trace distance. Then
\begin{eqnarray}
	x&=& \|\rho_A-\xi_A\|_1 = | {\rm tr}(A (\rho_0 -\rho))|\nonumber\\
  &=& | {\rm tr}(A(-t) (\rho_0(t) -\rho))|,
	\end{eqnarray}	
  where we used that the Gibbs state is invariant under the time evolution of its Hamiltonian.
	Applying the triangle inequality, the condition of strong localisation, and the fact that
	$e^{itH_A^l}A e^{-itH_A^l}$ is an observable supported on $B$ with operator norm upper bounded by unity,
	we get
	\begin{eqnarray}
	x&\leq& | {\rm tr}((e^{-itH}A e^{itH} - e^{-itH_A^l}A e^{itH_A^l} ) (\rho_0(t) -\rho))|\nonumber\\
	&+&
	  | {\rm tr}((e^{-itH_A^l}A e^{itH_A^l} ) (\rho_0(t) -\rho))|\nonumber\\
	 &\leq& \| (e^{-itH}A e^{itH} - e^{-itH_A^l}A e^{itH_A^l} ) \| \nonumber\\
	 &+& \|\rho_0(t)_B - \rho_B\|_1\nonumber\\
	 &\leq& c_\mathrm{loc} e^{- \mu l}+ \|\rho_0(t)_B - \rho_B\|_1,
\end{eqnarray}	
for $l$ suitably chosen. From this the statement follows.
\end{proof}
\end{document}